\documentclass[conference]{IEEEtran}
\IEEEoverridecommandlockouts
\usepackage{cite}
\usepackage{amsmath,amssymb,amsfonts}
\usepackage{multicol}
\usepackage{cleveref}
\usepackage{graphicx}  
\usepackage{textcomp}
\usepackage{xcolor}
\usepackage{theorem}
\usepackage{tabularx}
\usepackage{algorithm}
\usepackage[noend]{algpseudocode}
\usepackage{mathrsfs}
\usepackage{url}
\usepackage{subfigure}
\usepackage{tabularx}
\usepackage{float}
\usepackage{enumitem}
\usepackage{setspace}
\usepackage{booktabs}
\usepackage{threeparttable}
\usepackage{thmtools} 
\usepackage[numbers,sort&compress]{natbib}

\floatname{algorithm}{}

\usepackage[subtle,tracking=normal, wordspacing=normal,mathdisplays=tight]{savetrees}

\def\BibTeX{{\rm B\kern-.05em{\sc i\kern-.025em b}\kern-.08em
    T\kern-.1667em\lower.7ex\hbox{E}\kern-.125emX}}

 \newtheorem{lemma}{Lemma} \newtheorem{proof}{Proof}[section]
\renewenvironment{proof}{{\indent\it Proof:\quad}}{\hfill $\square$\par}
\makeatletter
\newcommand{\linebreakand}{%
  \end{@IEEEauthorhalign}
  \hfill\mbox{}\par
  \mbox{}\hfill\begin{@IEEEauthorhalign}
}
\makeatother

\begin{document}

\title{FileInsurer: A Scalable and Reliable Protocol for Decentralized File Storage in Blockchain
\thanks{This work has been performed with support from the National Natural Science Foundation of China (No. 11871366), Qing Lan Project of Jiangsu Province, China, and the Algorand Foundation Grants Program}
}

\author{\IEEEauthorblockN{Hongyin Chen\IEEEauthorrefmark{1}\IEEEauthorrefmark{2}, Yuxuan Lu\IEEEauthorrefmark{1}\IEEEauthorrefmark{3}, Yukun Cheng\IEEEauthorrefmark{4}\IEEEauthorrefmark{5}}
\thanks{\IEEEauthorrefmark{1} These authors contributed to the work equllly and should be regarded as co-first authors.}
\thanks{\IEEEauthorrefmark{4} Yukun Cheng is the corresponding author.}
\IEEEauthorblockA{\IEEEauthorrefmark{2}\IEEEauthorrefmark{3}\textit{Center on Frontiers of Computing Studies, Peking University, Beijing, China}}
\IEEEauthorblockA{\IEEEauthorrefmark{4}\textit{Suzhou University of Science and Technology, Suzhou, China}}
\IEEEauthorblockA{Email: \IEEEauthorrefmark{2}chenhongyin@pku.edu.cn, \IEEEauthorrefmark{3}yx\_lu@pku.edu.cn, \IEEEauthorrefmark{5}ykcheng@amss.ac.cn}\\
}

\maketitle
\begin{abstract}

With the development of blockchain applications, the requirements for file storage in blockchain are increasing rapidly. Many protocols, including Filecoin, Arweave, and Sia, have been proposed to provide scalable decentralized file storage for blockchain applications. However, the reliability is not well promised by existing protocols. Inspired by the idea of insurance, we innovatively propose a decentralized file storage protocol in blockchain, named as FileInsurer, to achieve both scalability and reliability. While ensuring scalability by distributed storage, FileInsurer guarantees reliability by enhancing robustness and fully compensating for the file loss. Specifically, under mild conditions, we prove that no more than 0.1\% value of all files should be compensated even if half of the storage collapses. Therefore, only a relatively small deposit needs to be pledged by storage providers to cover the potential file loss. Because of lower burdens of deposit, storage providers have more incentives to participate in the storage network. FileInsurer can run in the top layer of the InterPlanetary File System (IPFS), and thus it can be directly applied in Web 3.0, Non-Fungible Tokens, and Metaverse.

\end{abstract}


\section{Introduction}

File storage is a fundamental issue in distributed systems. Recently, the developments of Web 3.0~\cite{alabdulwahhab2018web}, Non-Fungible Tokens (NFTs)~\cite{wang2021non,chenabsnft}, and Metaverse~\cite{ryskeldiev2018distributed} have raised high requirements on reliability and accessibility of file storage. For example, the metadata of NFTs should be verifiable and accessible in NFT markets, as the values of NFTs disappear if the metadata is lost. Billions of metadata generated by blockchain applications are searching for reliable storage services.

Traditionally, people store files in personal storage or cloud storage service. However, personal storage struggles to keep files secure and accessible. Additionally, cloud storage lacks transparency and trust~\cite{benisi2020blockchain}. It is hard for users to recognize how many backups of their files should be stored to guarantee security. Moreover, file loss often occurs in cloud storage.

Due to the defects of personal storage and cloud storage, more and more users choose to store files in the blockchain-based \emph{Decentralized Storage Networks} (DSNs) such as Sia~\cite{vorick2014sia}, Filecoin~\cite{benet2018filecoin}, Arweave~\cite{williams2019arweave}, and Storj~\cite{wilkinson2014storj}. In a DSN, storage providers contribute their available hard disks to store files from clients and then earn profits. The storing, discarding, and storing state-changing events of files are recorded in the blockchain. Files can be stored by multiple storage providers to enhance security. Additionally, in Filecoin, backups are changed to be replicas, once they have been proved by proof-of-replication (PoRep). PoRep well resists Sybil attacks~\cite{douceur2002sybil} by a storage provider, who may pretend to store multiple backups by forging multiple identities, while she actually only stores one backup.
PoRep can also be used to ensure providers cannot cheat on the available storage space.

However, the existing DSN protocols do not well promise the reliability. In DSN, there are always files only stored by a small part of storage providers due to the issue of scalability. Therefore, it is impossible to completely avoid loss of files. When files are lost, the owners of these files only receive little compensation.

In this paper, we aim to enhance the reliability of decentralized file storage from the perspective of economic incentive approaches. For the Bitcoin Blockchain~\cite{nakamoto2008bitcoin}, the most success is to apply the economic incentive approach, by awarding a certain amount of token to encourage miners to actively mine. Thus, in the era of blockchain, the issue of economic approaches is getting more and more important. We build a decentralized insurance scheme on files stored in DSN to protect the interests of users when their files are lost. Under the insurance scheme, storage providers need to pledge a deposit before storing files. If a file is lost, which means that all providers storing this file are corrupted, the total deposit from these providers can fully compensate for the loss of this file.

We hope that the deposit should be small to incentivize participants to contribute their storage space. Let us denote \emph{deposit ratio} to be the ratio of the sum of deposits to the total value of files. Chen et. al. \cite{chen2020decentralized} firstly studied how to decrease the deposit ratio in the decentralized custody scheme with insurance. However, the methodology in \cite{chen2020decentralized} cannot be directly applied in our scenario. The reason is that storage providers and files change over time in DSN, while ~\cite{chen2020decentralized} is only suitable for static setting. Our approach is to achieve provable robustness by ensuring \emph{storage randomness}. Storage randomness requires the locations of replicas are randomly selected by DSN, such that these locations are 
uniformly distributed. Consequently, the attackers must corrupt a considerable portion of providers even if they only want to destroy all backups of a small portion of files. Therefore, the randomness can promises that only a relatively small deposit needs to be pledged by storage providers to cover the potential file loss.

\subsection*{Main Contributions}
We propose FileInsurer, a novel design for blockchain-based \emph{Decentralized Storage Network}, to achieve both scalability and reliability of file storage. In our protocol, storage providers are required to pledge deposits to registered sectors and the locations of files are randomly selected. To further ensure storage randomness, locations of files' replicas shall change from time to time because the list of sectors is dynamic.

Our protocol advances the technology of decentralized file storage in the following three aspects. 
\begin{itemize}
    \item Firstly, FileInsurer supports dynamic content stored in sectors with low cost, which is necessary to ensure \emph{storage randomness}. FileInsurer deploys \emph{Dynamic Replication} (DRep) to support adding and refreshing stored files. DRep is also able to resist Sybil attacks and make sure the free space of sectors is indeed available.
    \item Secondly, FileInsurer can achieve provable robustness. In FileInsurer, files are stored as replicas in sectors. Naturally, a file is missing, if and only if all replicas of this file have been destroyed. A sector is collapsed,  as long as any bit in this sector is lost. Under mild conditions, we prove that no more than 0.1\% value of all files are lost even if half of the storage collapses.
    \item Thirdly, FileInsurer implements an insurance scheme on DSN that can provide full compensation for the loss of those missing files. The compensation is covered by the deposit of all crashed storage sectors. Our theoretical analysis indicates that only a small deposit ratio is needed to cover all of the file loss in FileInsurer. 
\end{itemize}
To the best of our knowledge, FileInsurer is the first DSN protocol that can provide full compensation for the file loss and has provable robustness.

\subsection*{Paper Organization}
The rest of this paper is organized as follows. Section~\ref{related} introduces the related works of decentralized storage protocols. In Section~\ref{prelimiaries}, we describe the structure and components of FileInsurer protocol. Then, we continue to introduce the protocol design of FileInsurer in Section~\ref{protocoldesign}. In Section~\ref{analysis}, we propose the theoretical analysis on the scalability, robustness, and deposit issue of our protocol. We also compare FileInsurer with other blockchain-based decentralized storage protocols. In addition, some practical problems in FileInsurer are detailedly discussed in Section~\ref{discussion}. Finally, we summarize our protocol and raise some open problems in Section~\ref{conclusion}.

\section{Related Works}\label{related}
\subsection{InterPlanetary File System (IPFS)}
The InterPlanetary File System (IPFS) is a peer-to-peer distributed file system that seeks to connect all computing devices with the same system of files~\cite{benet2014ipfs}. Files, identified by their cryptographic hashes, are stored and exchanged by nodes in IPFS. Nodes also provide the service of retrieving files to earn profits through BitSwap protocol. The routing of IPFS is achieved by Distributed Hash Tables (DHTs), which is an efficient way to locate data among IPFS nodes. Based on BitSwap and DHTs, IPFS builds an Object Merkle DAG which allows participants to address files through IPFS paths.

\subsection{FileCoin}\label{related:filecoin}
Filecoin builds a blockchain-based \emph{Decentralized Storage Network} which runs in the top layer of IPFS~\cite{benet2018filecoin}. There are three types of participants in Filecoin, which are clients, storage miners, and retrieval miners. Specifically, clients pay to store and retrieve files, storage miners earn profits by registering sectors to offer storage, and retrieval miners earn profits by serving data to clients.

\subsubsection{Proof-of-Replication~(PoRep)}
PoRep~\cite{benet2017proof} is a kind of proof-of-storage scheme deployed in Filecoin. In the PoRep scheme, the prover firstly generates a replica of file $\mathcal{D}$, denoted by $\mathcal{R}^{\mathcal{D}}_{ek}$, through the process of \textsf{PoRep.setup($\mathcal{D}, ek$)}. $ek$ is a randomly chosen encryption key that with $ek$, $\mathcal{R}^{\mathcal{D}}_{ek}$ can be encrypted from $\mathcal{D}$, and $\mathcal{D}$ can be decrypted from $\mathcal{R}^{\mathcal{D}}_{ek}$. The prover then submits the hash root of $\mathcal{R}^{\mathcal{D}}_{ek}$ to the DSN. Finally, the prover proves that $\mathcal{R}^{\mathcal{D}}_{ek}$ is a replica of $\mathcal{D}$ with encryption key $ek$ via SNARK. 

The verification of SNARK is very efficient. However, the calculation of $\mathcal{R}^{\mathcal{D}}_{ek}$ would take a lot of time because it can't be parallelized. Additionally, the calculation of SNARK would consume lots of computation resources.

\subsubsection{Filecoin Sectors}
In Filecoin, sectors are divided into sealed ones and unsealed ones\footnote{See in \url{https://spec.filecoin.io/systems/filecoin_mining/sector/}}. Only sealed sectors are part of the Filecoin network and can get rewards of storage. Unsealed sectors only contain raw data, and a sealed sector can be registered from an unsealed sector by PoRep. Storage miners would pledge deposits when registering a sector, but when the sector crashes, that deposit is burnt other than used for compensating the file loss to clients.

When registering an unsealed sector, if the sector is not full, the rest space of the sector would be filled with zeros before encoding by PoRep. If a sealed sector doesn't contain any files, which means the contents of that sector are all zeros when registering, it's called a committed capacity (CC). Other sealed sectors are called regular sectors. A CC sector can be upgraded to a regular sector by discarding the CC sector and registering a new regular sector. However, the content of a regular sector can be no longer changed. 

\subsubsection{Proof-of-Spacetime}
PoSt is another kind of proof-of-storage scheme for storage miners to prove that they are indeed actually storing a replica. There are two kinds of PoSt in Filecoin. WinningPoSt serves as a part of the Expected Consensus of Filecoin, while WindowPoSt guarantees that the miner continuously maintains a replica over time. Therefore, Sybil attacks are prevented by the combination of WindowPoSt and PoRep because storage miners should actually store all replicas.

\subsubsection{Storage Market and Retrieval Market}
There are two markets in Filecoin, the Storage Market and the Retrieval Market. In the Storage Market, storage miners and clients negotiate on the price and length of storage. Similarly, retrieval miners and clients would negotiate on the price of file retrieving.

\subsection{Other Solutions to Decentralized File Storage}

\subsubsection{Storj}
Storj~\cite{wilkinson2014storj} is a sharding~\cite{kokoris2018omniledger,zhang2020cycledger} based protocol to archive a peer-to-peer cloud storage network implementing end-to-end encryption. It stores files in encrypted shards to ensure that the file itself cannot be recovered by anyone other than the owner. Moreover, it uses erasure code to ensure file availability in case some shards are lost.

\subsubsection{Sia}
Sia~\cite{vorick2014sia} is a platform for decentralized storage enabling the formation of storage contracts between peers. The Sia protocol provides an algorithm of storage proof in order to build storage contracts. According to the file contract, storage providers need to generate proof-of-storage periodically. The client needs to pay for each valid storage proof.

\subsubsection{Arweave}
Arweave~\cite{williams2019arweave} is a mechanism design-based approach to achieving a sustainable and permanent ledger of knowledge and history. Storing files on Arweave only requires a single upfront fee, after which the files become part of the consensus. Arweave uses the mechanism of Proof of Access in consensus to ensure that miners need to store as many files as possible to participate in mining.

\section{preliminaries}\label{prelimiaries}

FileInsurer is a protocol to build a blockchain-based Decentralized Storage Network (DSN)~\cite{benisi2020blockchain}. The structure of DSN could be an independent blockchain or a decentralized application (DApp) parasitic on existing blockchains or other distributed network types. In DSN, a group of participants, called {\em storage providers}, are willing to rent out their unused hardware storage space to store the {\em clients'} files, and then the distributed file storage is realized.

In this section, we introduce the structure and components of FileInsurer protocol. Particularly, we deploy \emph{Dynamic Replication} (DRep) to support dynamic content in sectors
with low cost, which is important to ensure \emph{storage randomness}. Additionally, we also explain why compensation is necessary for DSN.

\subsection{Participants}
There are two kinds of participants in DSN that are \textit{clients} and \textit{storage providers}.

\subsubsection{Client}
 Clients are the participants who have the demand to store files in the network. They propose a request to declare which file needs to be stored, via \textsf{File\_Add} request. Once her file is stored, the client shall pay the rent for the storage service at periodic intervals (introduced in Section \ref{fee_mechanism}), which depends on the file's value and size. 
 They also can ask DSN to discard their files stored before, via \textsf{ File\_Discard} request. Besides, clients can retrieve any file stored in DSN, via \textsf{File\_Get} request, by paying the retrieving payment. As the uploaded files are public in DSN, clients can encrypt their files before uploading if she concerns about privacy.

\subsubsection{Storage Providers}
Storage Providers are the participants who rent out their hard disks to store clients' files and offer the service of retrieving files in exchange for payments. 
When receiving the \textsf{File\_Add} request from a client, DSN
automatically selects several independent storage providers to store this file, so that the robustness could be guaranteed by replicating files.
After receiving a file from a client, providers need to declare that they have obtained this file by \textsf{File\_Confirm} request. In addition, after storing a file, it is necessary for providers to repeatedly submit the proofs of file storage to DSN at each specified checkpoint, to show that they are storing this file, via \textsf{File\_Prove} request. In order to guarantee security, each storage provider must pledge a deposit, so that her deposit could be liquidated to compensate for the loss of clients once her disk is corrupted. When a client requests retrieval of a specified file, the providers, who store this file, compete to respond to the request for the corresponding payment. Hence a {\em Retrieval Market} is formed, in which the clients and providers exchange the file without the witness of DSN.

\subsection{Data Structures}
\Cref{fig:datastructure} shows a brief description of the data structures of the FileInsurer. There are four main data structures, which are sector, file descriptor, allocation table, and pending list.
\subsubsection{Sector}
A disk sector is the smallest unit that a provider rents out to store files. Sector sizes vary but are required to be an integer multiple of a minimum value of $minCapacity$. $minCapacity$ can be set to $64$GB or other deterministic value. A sector is considered to be corrupted, as long as any bit in this sector is destroyed. A file is missing, if and only if the sectors storing this file are all corrupted. In FileInsurer protocol, providers could divide their storage spaces into multiple sectors, and are not allowed to register multiple disks as the same sector. In addition, FileInsurer requires that a file is integratedly stored in a sector, instead of being dispersed into multiple sectors.  
Such a requirement ensures the owner of a lost file obtains the compensation, which can completely make up for her loss. 

\subsubsection{File descriptor}
The file descriptor $f$ describes a file stored in the network, including its size, value, Merkle root, the number of copies, and other necessary information. When a file is stored, the following two conditions must be satisfied.
\begin{itemize}
    \item The total size of the files stored in a sector must not exceed the capacity of this sector.
    \item If a file $f$ is lost, meaning all sectors storing it are all corrupted, then the deposits from these sectors are at least $f.value$ to make up the loss of the file's owner.
\end{itemize}

\subsubsection{Allocation table}
FileInsurer selects some feasible sectors to store a file and makes a note of it recorded in the allocation table. The allocation table will be updated when a file is stored in the network, a file is discarded, or the storage location of a file is transferred. The allocation table is a part of the network consensus and can support fast random access.

\subsubsection{Pending list}
In the design of FileInsurer, some tasks need to be automatically executed at a specific time in the future, such as regularly checking whether a file is saved correctly. Therefore FileInsurer needs to maintain a pending list to save these tasks and their corresponding execution time. When a new time point $t$ is reached, the tasks in the pending list whose timestamp is $t$ will be automatically executed by the network. As the gas fee for these tasks should be paid in advance, tasks that are placed in the 
pending list must have a clear gas used upper bound. In the basic design of FileInsurer, these tasks are only generated through network consensus.

\begin{figure}[ht]
\begin{algorithm}[H]
\caption{Data structures}
\small
\setstretch{.9}
\textbf{Sector}\\
$\textsf{sector}: (\textsf{owner},~ \textsf{id},~ \textsf{capacity},~ \textsf{freeCap},~ \textsf{state})$
\begin{itemize}
    \item \textsf{owner}: the provider who owns the sector.
    \item \textsf{id}: the id of the sector, a provider cannot have two sectors with the same \textsf{id}.
    \item \textsf{capacity}: the storage capacity of the sector.
    \item \textsf{freeCap}: current free capacity of the sector.
    \item \textsf{state}: \texttt{normal} means this sector has capacity to accept new files, \texttt{disable} means the sector no longer accepts new files.
\end{itemize}

\textbf{File descriptor}\\
$\textsf{fileDescriptor}: (\textsf{size},~\textsf{value},~\textsf{merkleRoot},~\textsf{cp},~\textsf{cntdown},~\textsf{state})$
\begin{itemize}
    \item \textsf{size}: the size of the file.
    \item \textsf{value}: the value of the file.
    \item \textsf{merkleRoot}: merkle root of the file.
    \item \textsf{cp}: the number of replicas to be stored in the network, determined by the file value.
    \item \textsf{cntdown}: the number of checkpoints until the next refresh of the file store.
    \item \textsf{state}: \texttt{normal} means this file needs to be stored, \texttt{discard} means this file is discarded.
\end{itemize}

\textbf{Allocation table}\\
$\textsf{allocTable}: \left\{\left(\textsf{fileDescriptor},~\textsf{index}\right)\rightarrow \textsf{allocEntry}\right\}$

$\textsf{allocEntry}: (\textsf{prev},~\textsf{next},~\textsf{last},~\textsf{state})$
\begin{itemize}
    \item \textsf{prev}: the current sector storing the file.
    \item \textsf{next}: the next sector to store the file.
    \item \textsf{last}: time of the last proof of storage.
    \item \textsf{state}: \texttt{alloc} means the file is being (re)allocated to a sector, \texttt{confirm} means that the file is confirmed by the next sector to store, \texttt{normal} means the current sector is storing the file, \texttt{corrupted} means the current sector is corrupted.
\end{itemize}

\textbf{Pending list}\\
$\textsf{pendingList}: \left\{\textsf{time}\rightarrow \left[\textsf{task},\textsf{task},...\right]\right\}$
\begin{itemize}
    \item \textsf{time}: time point when the tasks need to be automatically executed.
    \item \textsf{task}: description and parameters of the task to be executed.
\end{itemize}

\end{algorithm}
\caption{\textbf{The data structures of FileInsurer}}
\label{fig:datastructure}
\end{figure}

\subsection{Interactions between Participants and Network} 

This subsection introduces the abovementioned operations performed by clients and storage providers in detail.

\subsubsection{Client requests}
\begin{itemize}
    \item \textsf{File\_Add}: \textit{Client stores a file in DSN.}\\
    A client submits an order through a \textsf{File\_Add} request to inform DSN of the file's description $f$, containing size $f.size$, value $f.value$, Merkle root $f.merkleRoot$, the number of replicas $f.cp$, and other necessary information. DSN automatically allocates feasible $f.cp$ sectors. When these sectors are found, the client transmits the file to these sectors.
    \item \textsf{File\_Discard}: \textit{Client discards a file stored in DSN.}\\
     It is not necessary for clients to specify how long to store the file in advance. As an alternative, the client can discard the file at any time by submitting \textsf{File\_Discard} request, which contains the description $f$ of this file, to DSN.
    \item \textsf{File\_Get}: \textit{Client retrieves a file from DSN.}\\
    Each client can request any file in DSN, via \textsf{File\_Get} request, by paying a certain amount of tokens. Because this requested file is available in multiple providers' sectors, the retrieve request can be satisfied by receiving one of the copies from these providers.
\end{itemize}

\subsubsection{Provider requests}

\begin{itemize}
    \item \textsf{Sector\_Register}: \textit{Provider registers a new sector in DSN.}\\
    When providers launch a new storage space, they have two options. One is to register the whole storage space as one sector. The other is to divide this space into several parts and each part is registered as one sector. When a sector is registered, the provider shall pledge a deposit proportional to the capacity of this sector.
    
    \item \textsf{Sector\_Disable}: \textit{An operation to affirm that a sector no longer accepts new files.}\\
    In the design of FileInsurer protocol, providers are not allowed to revoke the sectors they leased on the network before. Instead, when a provider decides not to provide storage service from a sector, she shall declare that the sector is disabled, that it is no longer accepts any new file. After all files stored in this sector are allocated to other sectors by the network, the sector is removed from the network.
    
    \item \textsf{File\_Confirm}: \textit{The provider confirms to the network that a file has been received.}\\
    The network automatically specifies the storage sector for files, and the provider of the sector needs to confirm to the network after receiving the client's file.
    
    \item \textsf{File\_Prove}: \textit{The provider submits the certificate to the network of its correct storage of files.}\\
    When providers store the files, they must repeatedly submit proofs of replication to ensure they are storing the files. Proofs are posted on and verified by DSN.
    
    \item \textsf{File\_Supply}: \textit{The provider responds a \textsf{File\_Get} request from a client.}\\
    Once the supply and demand relationship of one file has been established, the transmission of this file would be carried off-chain.
\end{itemize}

\subsection{Dynamic Content in Sectors}\label{PoRep}
In FileInsurer, the content of a sector needs to be dynamic from time to time, which is supported by ensuring storage randomness. FileInsurer can resist Sybil attack by storing the files as multiple replicas. In FileInsurer, these replicas are generated by PoRep, and the free capacity of a sector needs to be proven that it is indeed available. A trivial idea is to make a new replica of the sector whenever the content is changed by PoRep. However, it is not a wise solution because it would lead to an extremely high burden on providers and much more verification of PoRep.

We propose a novel solution called \emph{Dynamic Replication} (DRep) to solve this problem. Different from Filecoin, we don't encode a whole sector into a replica, but make each file in a sector to be a unique replica. We define a \emph{Capacity Replica} (CR) as a replica of zeros bits generated by the PoRep process. When a sector is registered, it should be just filled with $l$ unique CRs. The sector is requested to contain as many CRs as possible while storing files. Therefore, the unsealed space of a sector is smaller than the size of a CR. Figure~\ref{fig:dcs} shows some examples of DRep.

\begin{figure}[h]
    \centering
    \includegraphics[width=0.6\linewidth]{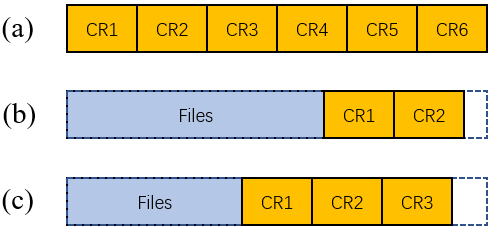}
    \caption{\textbf{Examples of DRep}: Initially the sector contains six Capacity Replicas (as shown in (a)). After filling some files, there are two Capacity Replicas left (as shown in (b)). When the total size of files decreases, the provider regenerates the $CR3$ (as shown in (c)).}
    \label{fig:dcs}
\end{figure}

Ensuring free space of sectors is indeed available by CRs is an efficient way. All CRs only need to be generated by PoRep once and then verified continuously stored via WindowPoSt~\cite{benet2018filecoin}. If a CR has been thrown, the provider can recover it by $\textsf{PoRep.setup}$ because the raw data of a CR are zeros. It doesn't need to go through the whole PoRep process because the Merkle roots of CRs have been previously verified. Therefore, DRep won't bring an extra verification burden on the DSN and providers don't need to generate SNARK of PoRep again.

Additionally, FileInsurer changes the location of replicas at a low cost. Consider that a replica of a file $f$ needs to be transferred to another sector. The provider do not need to generate new replicas of $f$ by PoRep, but just transfer the old ones. Liveness issue occurs that a provider may not transfer the replica of $f$ to the successor provider. However, it doesn't bother because the successor provider can fetch the source data of $f$ from other providers and recover the replica via \textsf{PoRep.Setup}. Similar to CRs, these replicas don't need to be verified again, and they can be recovered from the raw file. Therefore, the movement of replicas is efficient.

\subsection{Storage Market and Retrieval Market}
Similar to Filecoin, there are two markets in FileInsurer, the Storage Market and the Retrieval Market.
The Retrieval Market in FileInsurer is the same as that in Filecoin. In the DSN, clients can send the request to retrieve any file $f$. Any participant with $f$ or a replica of $f$ can answer that request. The process of retrieving is accomplished by BitSwap of IPFS.
However, the Storage Market in FileInsurer is quite different from the one in Filecoin. In FileInsurer, the price of storing a file is decided by the size and the value of that file. Clients do not need to negotiate prices of storage with providers and even do not need to specify who to store their files. Prices for storage services may change over time, which will be discussed in Section~\ref{fee_mechanism}.

\subsection{Source of Randomness}
Just like other DSN protocols, FileInsurer needs a huge amount of on-chain random bits. To achieve this with low expense, we use a pseudorandom number generator~\cite{james1990review,pareek2014overview} to generate long pseudo-random bits based on a short random beacon. Additionally, the issue of generating an unbiased and unpredictable public random beacon in blockchain has been well studied~\cite{cachin2005random,bhat2021randpiper,das2021spurt}. Combining the abovementioned two technologies, we can cheaply get enough public pseudo-random bits. In this paper, we omit the implementation of generating and using random bits because it is too detailed and not our main contribution.

\subsection{Necessity of Compensation in DSN}\label{comneeded}
Compensation is needed in DSN because the scalability of DSN would lead to an unavoidable risk of missing data. Necessarily, the scalability of storage means that a participant in DSN only needs to store a very small part of all data in DSN. Therefore, many data of DSN must only be stored by a small part of participants. 
However, if a constant ratio, for example, $0.1$, of sectors (or storage capacity) crash instantaneously, some data may be lost. So DSN brings a huge risk of file loss.

To demonstrate more clearly, let us provide some other concrete examples. In Storj, a file is lost if enough shards of the file are not available beyond what can be recovered by erasure code. In Filecoin, a file is lost, if and only if all sectors storing replicas of this file crash down.

To balance the safety and the scalability of DSN, compensation is an effective method to motivate users to take part in the distributed file storage. The reasonable deposit shall compensate the users' loss from missing data.

\section{Protocol Design of FileInsurer}\label{protocoldesign}
In this section, we introduce the protocol design of FileInsurer in detail. The insurance scheme is introduced into the protocol design so that the storage providers are responsible for file loss and their deposit can fully compensate the clients whose files are lost. To support the dynamical file storing in sectors, storage randomness is needed to randomly distribute the locations of replicas in DSN, which can be realized by randomly selecting and refreshing the locations of replicas.
Additionally, files with higher values have more replicas so it is harder to destroy all replicas of these files.

In FileInsurer, all file replicas and Capacity Replicas are generated by PoRep, which means that WinningPoSt can be easily achieved. Therefore, the Expected Consensus deployed by Filecoin can be directly applied to our consensus algorithm. Additionally, FileInsurer protocol can be deployed as a smart contract or sidechain in other blockchain protocols such as Ethereum~\cite{wood2014ethereum} and Algorand~\cite{gilad2017algorand, chen2019algorand}.

\subsection{Fee mechanism}\label{fee_mechanism}
In our DSN design, clients need to pay a fee when they obtain the storage service and retrieval service. Moreover, there are three kinds of fees in FileInsurer, which are the traffic fee, storage rent, and prepaid gas fee.

\subsubsection{Traffic fee}
The traffic fee needs to be paid when a client occupies the network bandwidth of providers by transmitting files, retrieving files, or other interactions. The mechanism to pay a traffic fee is necessary because malicious clients may transmit files but pay nothing to block the providers' network. The operation to upload traffic fee must be committed to the storage provider before the file transmission, and the provider obtains the fee only when it has confirmed the file.

\subsubsection{Storage rent} Clients need to pay the storage rent for the used storage space, which is proportional to the size of the file times the number of replicas. The unit rent is the same for all files, and the network informs the client how much rent it should pay. The client will be automatically charged storage rent in the task \textsf{Auto\_CheckAlloc} which will be introduced in \cref{protocol}.
In particular, the network distributes revenue by time period. In a time period, all storage rent is stored in the network at first. At the end of the period, the network distributes the rent to owners of proper functioning sectors during this period. Storage providers are paid proportionally according to their total storage capacity, without paying attention to which file is stored in which sector.

\subsubsection{Prepaid gas fee} After a client stores files on the network, the network needs to periodically check the proof and refresh the file storage locations. These operations use the consensus space and thus incur a gas fee. The gas fee for these operations should be prepaid by the user as these operations are performed automatically. The prepaid gas fee shall be collected together with storage rent through
\textsf{Auto\_CheckAlloc}.

In addition, anyone who submits requests to the network must pay a gas fee to avoid wasting valuable consensus space. The design of the gas fee mechanism is part of the network design. As our DSN design does not focus on the network design, we can use other existing gas fee mechanisms and do not detailedly address it in this work.

\subsection{Deposit and Compensation}
When registering a sector, the storage provider should pledge to DSN with a certain amount of deposit. The deposit is locked until the sector safely quits the system or is corrupted. If the sector safely quits, the deposit would be withdrawn to the storage provider. If the sector is corrupted, the deposit must be confiscated.

When the deposit of a sector is confiscated, it shall be stored in the network to compensate for lost files. File loss in a network means that all its copies are no longer available, i.e. those storage sectors storing the copies are all corrupted. When a file is lost, the network shall provide users with compensation equal to the value of the file. Values of files are given by users when storing their files. If a user reports a higher value than the value of her file, she would pay a higher storage rent, and if she reports a lower value, the compensation would be lower once her file gets lost.

The \emph{deposit ratio} $\gamma_{deposit}$ of FileInsurer is defined as the ratio that the sum of deposits compared to the maximal value of files stored in the network. It can be understood as how much deposit is required for each unit of value stored in the network. Thus, the lower the deposit ratio makes the providers have more incentives to participate in the distributed storage network, and thus make our protocol more competitive.

Now we show how to calculate the deposit by $\gamma_{deposit}$ while registering a sector. Assume that the total size of sectors in FileInsurer is $N_s\times minCapacity$ and the maximal total value of stored files are $N_v^m \times minValue$. For a sector $s$ with capacity $s.capacity$, the deposit should be the proportion of $s.capacity$ in the network multiplied by the total deposit, which is 
$\gamma_{deposit} \times N_v^m \times minValue \times \frac{s.capacity}{N_s \times minCapacity}$. Let $capPara = \frac{N_v^m}{N_s}$ be a constant and the deposit becomes $ s.capacity \times \gamma_{deposit} \times\frac{ capPara \times minValue }{minCapacity}$, which can be calculated only by $s.capacity$, $\gamma_{deposit}$, and some constants. The setting of $\gamma_{deposit}$ are discussed in \Cref{th:ratio}.

\subsection{Main Protocol}
\label{protocol}
\begin{figure*}[htbp]
    \centering
    \subfigure[\textbf{Storing files on the network}: First, the file should be informed to the network to get the sectors where the files are stored at. The client then sends the file to those sectors. The file is successfully stored after the system executes \textsf{Auto\_CheckAlloc}, and the rent is paid every time the system executes \textsf{Auto\_CheckProof}.]
    {
    \includegraphics[scale=.4]{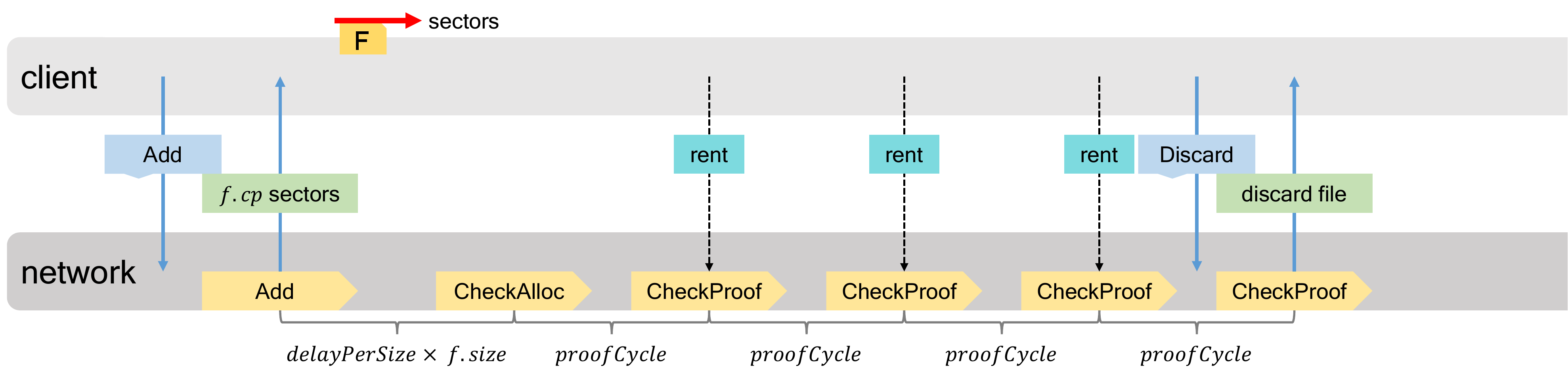}
    }
    \subfigure[\textbf{Renting sectors to the network}: After the sector has been registered, the file will be swapped into or out of the sector through \textsf{Auto\_Refresh} from time to time. Moreover, the network may also inform the sector to take over new files by corresponding \textsf{File\_Add}.]
    {
    \includegraphics[scale=.4]{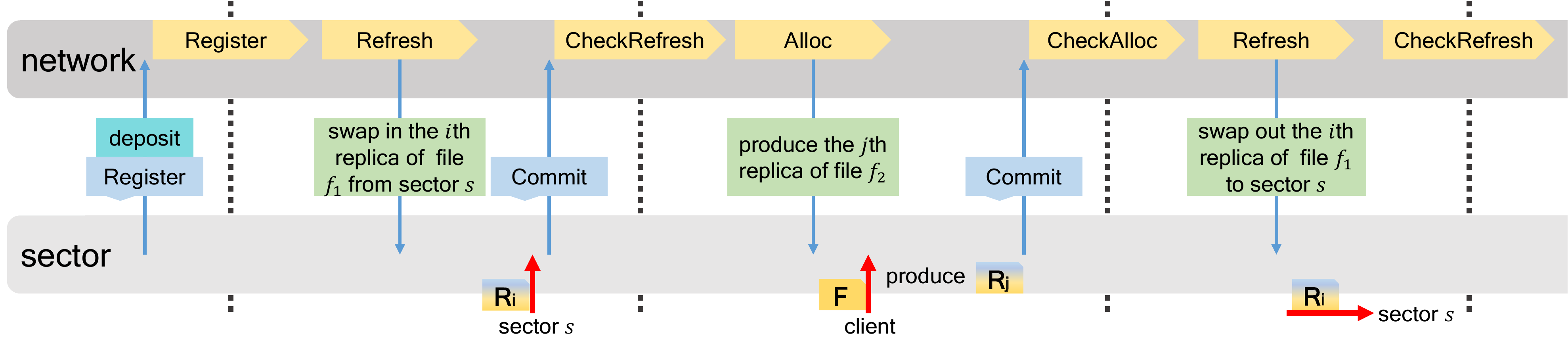}
    }
    \caption{\textbf{Brief overview of the protocol}: How files and sectors interact with the network. The symbol ``F'' represents a file, and ``R'' represents a file replica.}
    \label{brief}
\end{figure*}

\begin{table}[htbp]
\setstretch{0.6}
\caption{Descriptions of parameters and functions}
\label{table:description}
{\centering
\small
\begin{tabular}{c c} 
\textbf{Notation} & \multicolumn{1}{p{0.6\columnwidth}}{\textbf{Description}} \\ [0.5ex] 
\toprule
$RandomSector()$ & \multicolumn{1}{m{0.6\columnwidth}}{Sample a random sector. The probability of selecting each sector is proportional to its capacity.}\\
\midrule
$SampleExp(x)$ & \multicolumn{1}{m{0.6\columnwidth}}{Sample from an exponential distribution with mean $x$.}\\
\midrule
$RandomIndex(f)$ & \multicolumn{1}{m{0.6\columnwidth}}{Sample a number between $1$ and $f.cp$ uniformly at random.}\\
\midrule
$DelayPerSize$ & \multicolumn{1}{m{0.6\columnwidth}}{The maximum transmit time allowed per unit file size. This constant multiplied by the file size is the upper limit of the file transfer time allowed by the network.}\\
\midrule
$AvgRefresh$ & \multicolumn{1}{m{0.6\columnwidth}}{The number of $ProofCycle$s to refresh the file storage on average} \\
\midrule
$ProofCycle$ & \multicolumn{1}{m{0.6\columnwidth}}{Time interval between each inspection proof.}\\
\midrule
$ProofDue$ & \multicolumn{1}{m{0.6\columnwidth}}{The specified upper limit of the time the last proof until now.}\\
\midrule
$ProofDeadline$ & \multicolumn{1}{m{0.6\columnwidth}}{The tolerable upper limit of the time the last proof until now.}\\
\bottomrule\\
\end{tabular}}
\end{table}

FileInsurer mainly includes three parts:
\begin{itemize}
    \item \textsf{File\_}: protocols with \textsf{File\_} prefix handles the storage of data on the network,
    \item \textsf{Sector\_}: protocols with \textsf{Sector\_} prefix handles the sector registration and revocation,
    \item \textsf{Auto\_}: protocols with \textsf{Auto\_} prefix are mainly used for the maintenance of network. They are special because they cannot be called by anyone and will be executed automatically at a specific time.
\end{itemize}

Figure \ref{brief} proposes a brief overview of the protocol of FileInsurer by explaining how files and sectors interact with the network. Table \ref{table:description} lists all parameters and functions used in FileInsurer protocol.

\begin{figure}[ht]
\begin{algorithm}[H]
\caption{\textsf{File} protocol: client part}
\footnotesize
\setstretch{.9}
\textsf{File\_Add}
\begin{itemize}
    \item \textit{Inputs.} the size of the file $sz$, the value of the file $val$ and the merkle root of the file $rt$
    \item \textit{Goal.} generate the file descriptor for the file and allocate $k$ sectors to it for storage
\end{itemize}
\begin{algorithmic}
\State $f\gets (size=sz,value=val,merkleRoot=rt,cp=backupCnt(val),cntdown=-1,state=\texttt{normal})$
\State $count\gets 0$
\For{$i \in [f.cp]$}
    \State $s\gets RandomSector()$
    \While {$s.freeCap< f.size$}
    \Comment{almost never happens}
        \State $s\gets RandomSector()$
    \EndWhile
    \State let $e$ be a reference of $\textsf{allocTable}[f,i]$
    \State $count \gets count+1$
    \State $e \gets (prev=\textbf{null},next=s,last=-1,state=\texttt{alloc})$
\EndFor
\State $t \gets Now+DelayPerSize\times f.size$
\State add $CheckAlloc(f)$ to $\textsf{pendingList}[t]$
\end{algorithmic}


\textsf{File\_Discard}
\begin{itemize}
    \item \textit{Inputs.} a file descriptor $f$
    \item \textit{Goal.} discard file $f$
\end{itemize}
\begin{algorithmic}
    \State $f.state \gets \texttt{discard}$
\end{algorithmic}

\end{algorithm}
\caption{\textbf{\textsf{File} Protocol}: Add and discard files}
\label{file:client}
\end{figure}
\subsubsection{File\_}
\Cref{file:client} shows the network response for the clients' \textsf{File\_} requests. When a client makes a \textsf{File\_Add} request, the network first generates a file descriptor and samples $f.cp$ sectors for storage. The probability of each sector being selected is proportional to the capacity of this sector. The number of backup files that need to be stored is calculated by $f.cp=\frac{f.value}{minValue}k$, where $minValue$ is a parameter representing the lower limit of the file value of network storage and each $f.value$ must be integer multiple of $minValue$. Next, the waiting time is calculated and the user needs to transfer the file to the owner of the selected sectors before the waiting time expires. Once the waiting time expires, a task named as \textsf{Auto\_CheckAlloc} is performed automatically to confirm whether the file is successfully stored on the network. When a client submits a \textsf{File\_Discard} request, the network simply sets the state of the corresponding file descriptor to \texttt{discard}.

\begin{figure}[ht]
\begin{algorithm}[H]
\caption{\textsf{File} protocol: provider part}
\footnotesize
\setstretch{.9}
\textsf{File\_Confirm}
\begin{itemize}
    \item \textit{Inputs.} file descriptor $f$, index $i$ and sector $s$
    \item \textit{Goal.} confirm that a selected sector begins to store a specific file
\end{itemize}
\begin{algorithmic}
\State check the request is from the owner of sector $s$
\State verify $\textsf{allocTable}[f,i].next = s$ and $\textsf{allocTable}[f,i].state = \texttt{alloc}$ 
\State $entry.state \gets \texttt{confirm}$
\end{algorithmic}


\textsf{File\_Prove}
\begin{itemize}
    \item \textit{Inputs.} file descriptor $f$, index $i$, sector $s$ and proof $\pi$
    \item \textit{Goal.} verify that a selected sector is storing specific file
\end{itemize}
\begin{algorithmic}
\State check the request is from the owner of sector $s$
\State verify $\textsf{allocTable}[f,i].prev = s$
\State verify $\pi$ is a valid proof at time $\pi.t$
\State $\textsf{allocTable}[f,i].last \gets \pi.t$
\end{algorithmic}
\end{algorithm}
\caption{\textbf{\textsf{File} Protocol}: Confirm and prove files}
\label{file:provider}
\end{figure}

\Cref{file:provider} illustrates the network response for providers' \textsf{File\_} requests. When receiving an \textsf{File\_Confirm} request, the network sets the state of the corresponding allocation entry to \texttt{confirm}.  It means the sector has successfully received the file. When the network receives a \textsf{File\_Prove} request, it shall update the last proof time of the file storage after checking the correctness of the proof.

\begin{figure}[ht]
\begin{algorithm}[H]
\caption{\textsf{Sector} protocol}
\footnotesize
\setstretch{.9}
\textsf{Sector\_Register}
\begin{itemize}
    \item \textit{Inputs.} capacity $cap$ and owner $own$
    \item \textit{Goal.} register a new sector on the network
\end{itemize}
\begin{algorithmic}
\State the owner pledges deposit proportional to the sector size
\State $s\gets (owner=own,id=nextId(own),capacity=cap,freeCap=cap,state=\texttt{normal})$
\end{algorithmic}
\textsf{Sector\_Disable}
\begin{itemize}
    \item \textit{Inputs.} sector $s$
    \item \textit{Goal.} mark the sector as disabled
\end{itemize}
\begin{algorithmic}
    \State $s.state \gets \texttt{disable}$
\end{algorithmic}
\end{algorithm}
\caption{\textbf{\textsf{Sector} Protocol}: Register and disable sectors}
\label{sector}
\end{figure}
\subsubsection{Sector\_}
It is simple for the network to respond to \textsf{Sector\_} requests. The pseudo-code is shown in Figure \ref{sector}. A new sector is registered when a \textsf{Sector\_Register} request is received and the state of a sector will be set to \texttt{disable} when a request of \textsf{Sector\_Disable} is received. 
When all files in a disabled sector are swapped out, then it can be removed.

\subsubsection{Auto\_}
Note that the tasks with \textsf{Auto\_} prefix cannot be called by anyone and shall be executed at a specific time automatically. In the design of the FileInsurer protocol, the network needs to maintain a pending list to ensure that these tasks are executed at a specific time. There are 4 kinds of tasks with \textsf{Auto\_} prefix, which are \textsf{Auto\_CheckAlloc}, \textsf{Auto\_CheckProof}, \textsf{Auto\_Refresh}, and \textsf{Auto\_CheckRefresh}. In simple terms, \textsf{Auto\_CheckAlloc} is used to check that the file has been correctly stored on the network, \textsf{Auto\_CheckProof} is periodically proof checking, while \textsf{Auto\_Refresh} and \textsf{Auto\_CheckRefresh} are the processes of file storage refreshing in order to ensure the randomness of storage. Therefore, the period of proof checking should be short, and thus the frequency of the file storage location refreshing could be very low.

\begin{figure}[ht]
\begin{algorithm}[H]
\caption{\textsf{Auto} protocol}
\footnotesize
\setstretch{.9}
\textsf{Auto\_CheckAlloc}
\begin{itemize}
    \item \textit{Inputs.} file descriptor $f$
    \item \textit{Goal.} check if file $f$ is already confirmed by all of the selected sectors
\end{itemize}
\begin{algorithmic}
\For{$i \in [f.cp]$}
    \State let $e$ be a reference of $\textsf{allocTable}[f,i]$
    \If{$e.state \neq \texttt{confirm}$ and $e.state \neq \texttt{corrupted}$}
        \State inform that the client failed to upload the file $f$
        \State remove $f$ from the network
    \EndIf
\EndFor
\For{$i \in [f.cp]$}
    \State let $e$ be a reference of $\textsf{allocTable}[f,i]$
    \If{$e.state = \texttt{confirm}$}
        \State $e \gets (prev=e.next,next=\textbf{null},last=Now,state=\texttt{normal})$
    \Else
        \State $e \gets (prev=\textbf{null},next=\textbf{null},last=-1,state=\texttt{corrupted})$
    \EndIf
\EndFor
\State $f.cntdown\gets SampleExp(AvgRefresh)$
\State add $CheckProof(f)$ to $\textsf{pendingList}[Now+ProofCycle]$
\State inform that the client succeed to upload the file $f$
\end{algorithmic}
\end{algorithm}
\caption{\textbf{\textsf{Auto\_CheckAlloc}}: Check each allocation has confirmed the file}
\label{auto:checkalloc}
\end{figure}
\textsf{Auto\_CheckAlloc} will be executed automatically at some time after a \textsf{File\_Add} request is responded by the network. The network shall confirm if all $f.cp$ sectors have received the file described by $f$. If so, the network goes to change the state of the file descriptor to \texttt{normal}; otherwise, it shall inform the client that it failed to upload the file.

\begin{figure}[ht]
\begin{algorithm}[H]
\caption{\textsf{Auto} protocol}
\footnotesize
\setstretch{.9}
\textsf{Auto\_CheckProof}
\begin{itemize}
    \item \textit{Inputs.} file descriptor $f$
    \item \textit{Goal.} check that all storage locations of file $f$ are working
\end{itemize}
\begin{algorithmic}
\If{the client of file $f$ has does not have enough tokens to pay the cost for the next cycle}
    \State $f.state \gets \texttt{discard}$
    \State inform that file $f$ is discarded due to insufficient cost
\EndIf
\If{$f.state = \texttt{normal}$}
    \State deduct the cost for the next cycle from the client's account
    \For{$i \in [f.cp]$}
        \State let $e$ be a reference of $\textsf{allocTable}[f,i]$
        \If{$e.prev$ is not corrupted}
            \If{$e.last < Now - ProofDeadline$}
                \State confiscate the deposit of $s$
                \State mark and inform that $s$ is corrupted
            \ElsIf{$e.last < Now - ProofDue$}
                \State punish $e.prev$
            \EndIf
        \EndIf
    \EndFor
\EndIf
\If{$f.state=\texttt{discard}$}
    \State remove $f$ from the network
\ElsIf{$\forall j, \textsf{allocTable}[f,j].prev$ is corrupted}
    \State inform that file $f$ is lost
    \State compensate to the client
    \State remove $f$ from the network
\Else
    \State add $CheckProof(f)$ to $\textsf{pendingList}[Now+ProofCycle]$
    \State $f.cntdown\gets f.cntdown-1$
    \If{$f.cntdown=0$}
        \State $i\gets RandomIndex(f)$
        \State \textbf{call} $Refresh(f,i)$
    \EndIf
\EndIf
\end{algorithmic}
\end{algorithm}
\caption{\textbf{\textsf{Auto\_CheckProof}}: Check each proof of the file}
\label{auto:checkproof}
\end{figure}
Every file needs to be checked at some specific time whether it is stored properly. In each specific time period, a task named \textsf{Auto\_CheckProof} automatically runs to check whether each proof to the file is timely. We provide the pseudo-code of \textsf{Auto\_CheckProof} in Figure \ref{auto:checkproof}. We use WindowPoSt of Filecoin~\cite{benet2018filecoin} to implement the proof process. A sector will be punished if it cannot submit the proof of storage of its files within $ProofDue$ time, and then its corresponding deposit is liquidated if the proof of storage of its files cannot be provided within $ProofDeadline$ time.

\begin{figure}[ht]
\begin{algorithm}[H]
\caption{\textsf{Auto} protocol}
\footnotesize
\setstretch{.9}
\textsf{Auto\_Refresh}
\begin{itemize}
    \item \textit{Inputs.} file descriptor $f$ and index $i$
    \item \textit{Goal.} change the $i$-th storage place of file $f$ to a random sector
\end{itemize}
\begin{algorithmic}
\State $s\gets RandomSector()$
\If {$s.freeCap\geq f.size$}
    \State $\textsf{allocTable}[f,i].next \gets s$
    \State $\textsf{allocTable}[f,i].state \gets \texttt{alloc}$
    \State $t \gets Now+DelayPerSize\times f.size$
    \State add $CheckRefresh(f,i)$ to $\textsf{pendingList}[t]$
    \State $pre\gets \textsf{allocTable}[f,i].prev$
    \State inform the $i$th replica of file $f$ should be swapped from $pre$ to $s$
\Else
    \Comment{almost never happens}
    \State $f.cntdown\gets SampleExp(AvgRefresh)$
\EndIf
\end{algorithmic}


\textsf{Auto\_CheckRefresh}
\begin{itemize}
    \item \textit{Inputs.} file descriptor $f$ and index $i$
    \item \textit{Goal.} check whether the last refresh for file $f$ is confirmed
\end{itemize}
\begin{algorithmic}
    \State let $e$ be a reference of $\textsf{allocTable}[f,i]$
    \If{$e.state=\texttt{confirm}$}
        \State $e \gets (prev=e.next,next=\textbf{null},last=Now,state=\texttt{normal})$
        \State $f.cntdown\gets SampleExp(AvgRefresh)$
    \Else
        \State punish $entry.next$
        \For{$j \in [f.cp]$}
            \State punish $\textsf{allocTable}[f,j].prev$
        \EndFor
        \State \textbf{call} $Refresh(f,i)$ 
    \EndIf
\end{algorithmic}
\end{algorithm}
\caption{\textbf{\textsf{Auto\_Refresh} and \textsf{Auto\_CheckRefresh}}: Swap in and out the files}
\label{auto:refresh}
\end{figure}
Whenever a random number\footnote{This random number follows an exponential distribution} of checkpoints are passed, a task named \textsf{Auto\_Refresh} will be called to randomly refresh one of the storage places of the file. Figure \ref{auto:refresh} shows the details of \textsf{Auto\_Refresh} and another corresponding task \textsf{Auto\_CheckRefresh}. The probability of sampling the new storage sector is proportional to the capacity of the sector. The network then calculates a waiting time and the current sectors that store this file need to transfer it to the selected sector before the waiting time expires. Once the waiting time expires, the task named \textsf{Auto\_CheckRefresh} will be executed automatically to confirm whether the file is successfully stored in the new sector.

\section{Analysis}\label{analysis}
In this section, we analyze the performance of our protocol and compare FileInsurer with other DSN protocols in detail.
\subsection{Notation and Assumption}
\begin{table}[htbp]
\setstretch{0.7}
\caption{Notation Table}
\label{table:notation}
{\centering
\small
\begin{tabular}{c c} 

\textbf{Notation} &  \multicolumn{1}{p{0.7\columnwidth}}{\textbf{Description}} \\ [0.5ex] 
\toprule
$minCapacity$ & \multicolumn{1}{m{0.7\columnwidth}}{The minimum capacity of a sector. The capacity of each sector is an integer multiple of $minCapacity$.}\\
\midrule
$minValue$ & \multicolumn{1}{m{0.7\columnwidth}}{The minimum value of a file. The value of each file is an integer multiple of $minValue$.} \\
\midrule
$N_f$ & \multicolumn{1}{m{0.7\columnwidth}}{The number of files. }\\
\midrule
$N_{s}$ & \multicolumn{1}{m{0.7\columnwidth}}{The ``weighted'' number of sectors. $N_s\times minCapacity$ indicates the total capacity of the network.}\\
\midrule
$N_v$ & \multicolumn{1}{m{0.7\columnwidth}}{The ``weighted'' number of files. $N_v\times minValue$ indicates the total value of files stored on the network.}\\
\midrule
$N_v^m$ & \multicolumn{1}{m{0.7\columnwidth}}{The maximum ``weighted'' number of files the network is designed to carry. $N_v^m\times minValue$ is the maximum value the network can carry.}\\
\midrule
$\gamma^m_v$ & \multicolumn{1}{m{0.7\columnwidth}}{$\gamma^m_v=\frac{N_v}{N_v^m}$ is the ratio that the total value stored in FileInsurer compared to the maximal value.}\\
\midrule
$\gamma_{deposit}$ & \multicolumn{1}{m{0.7\columnwidth}}{ The deposit ratio.
$\gamma_{deposit}$ is the ratio that the sum of
deposits compared to the maximal value of files stored in the network.}\\
\midrule
$capPara$ & \multicolumn{1}{m{0.7\columnwidth}}{$capPara$ is defined as $\frac{N_v^m}{N_s}$.}\\
\midrule
$c$ & \multicolumn{1}{m{0.7\columnwidth}}{Security parameter. We set it to be $10^{-18}$.}\\
\midrule
$k$ & \multicolumn{1}{m{0.7\columnwidth}}{The number of backups should be stored of a file whose value is $minValue$.}\\
\bottomrule\\
\end{tabular}}
\end{table}

Before the analysis, we list notations in Table~\ref{table:notation} which are necessary for our theoretical analysis. Additionally, The following assumptions are necessary for our theoretical analysis.
\begin{itemize}
    \item \textbf{Consensus security}: FileInsurer requires that the network consensus itself is secure. The issue of consensus security is not the target of this paper.
    \item \textbf{Adversary ability}: FileInsurer allows an adversary to corrupt $\lambda$ proportion of network capacity immediately.
    \item \textbf{Redundant capacity}: FileInsurer requires that the total capacity in the network is no less than twice the total size of all files' replicas. This assumption is deployed to ensure \emph{storage randomness}.
\end{itemize}

\subsection{Performance of FileInsurer}
\subsubsection{Analysis for Capacity Scalability}
We consider the capacity scalability of FileInsurer as the maximal size of stored files. The following theorem indicates that FileInsurer is scalable in capacity.

\begin{restatable}{theorem}{thscala}\label{th:scala}
The total size of files can be stored in FileInsurer is
{\small
\[
\min \left\{\frac{ N_s \times minCapacity}{2r_1 k}, \frac{N_s \times minCapacity }{r_2} \right\},
\]
}
where
{\small
\begin{align}
r_1 & = \frac{\sum_f f.size \times f.value}{minValue \times \sum_f f.size},\label{r1}\\
r_2 & = \frac{minCapcity \times \sum_f f.value}{minValue \times \sum_f f.size \times capPara}.\label{r2}
\end{align}
}
\end{restatable}

The proof of Theorem~\ref{th:scala} is provided in Appendix A.
We claim that each of $r_1$ and $r_2$ is bounded by a constant in Section~\ref{dis:r12}. Then the total size of raw files can be stored in FileInsurer is $\tilde O(N_s \times minCapacity )$, which is almost linear to the total size of sectors.

\subsubsection{Storage Randomness}
Storage randomness is an important issue in FileInsurer. Storage randomness can ensure the locations of replicas are evenly distributed. Therefore,
the adversaries must corrupt a huge number of sectors even if they only want to destroy all replicas of a small portion of files. In FileInsurer, replicas are stored by randomly selected sectors in \textsf{File\_Add} and their locations are randomly refreshed by \textsf{Auto\_Refresh}. Such operations make the locations of all replicas are independent and identically distributed. 

However, when the total used space is close to the capacity of DSN, the process of \textsf{File\_Add} and \textsf{Auto\_Refresh} faces the trouble that the free space of selected sectors is not enough for the storage of a replica. We call this event a \emph{collision}. Although sectors can be reselected to store these replicas, Storage randomness would be influenced. Therefore, redundant capacity is required to avoid collisions. We claim that the frequency of collisions is ignorant by preliminary theoretical proof and further experiments. 

We first consider a trivial case that all files have the same size. The following theorem indicates that a collision happens with an extremely low probability.
\begin{restatable}{theorem}{thadd}\label{th:adding}
If all files have the same size $f.size$, for a sector $s$ with total capacity $s.capacity$ and free capacity $s.freeCap$, then 
{\small\[
\Pr\left[ \exists s,~s.freeCap\leq \frac{1}{8}s.capacity\right] \leq N_s\exp\left\{-0.144\frac{s.capacity}{f.size}\right\}.
\]}
\label{th:random}
\end{restatable}

The proof of Theorem \ref{th:adding} is provided in Appendix B. 
By \Cref{th:random}, when $\frac{s.capacity}{f.size} \geq 1000$ and $N_s \leq 10^{12}$, we have $\Pr\left[ \exists s,~s.freeCap\leq \frac{1}{8}s.capacity\right] < 10^{-50}$.

A replica of the file can be stored in any sector $s$ with $s.freeCap\leq \frac{1}{8}s.capacity$ because $f.size < \frac{1}{8}s.capacity$. This result indicates that the probability of collision is extremely low under these conditions.

We further consider the general case that the size of files follows a certain distribution. We conduct a series of numerical experiments in two different settings. In the first setting, we reallocate all file backups in one go for $100$ times. In the second setting, we allocate each file backup and then randomly refresh the location of a file backup $100N_{cp}$ times. Recall that $N_{cp}=kN_v$ is the number of file backups and each file $f$ needs to store $f.cp$ backups on the network.

In the experiments, we test several distributions for the size of file backups. We focus on the maximum ratio of capacity usage. If the ratio is less than $1$, no file backups are allocated to sectors with insufficient capacity. \Cref{table:experiment} shows the results of our experiments. We can find that the maximum ratios of capacity usage never exceed $0.64$ under all tested distributions, which means that the probability that file backups are allocated to sectors with insufficient capacity is very small. Therefore, the results of our experiments indicate that collisions would hardly occur when the average size of file backups is much smaller than the sector capacity.

We also discuss how to maintain storage randomness when the list of sectors changes in \cref{maintainrandom}. These results show that storage randomness is easy to be promised in practice. Therefore, each allocation of replicas is assumed to be independent and identically distributed in the following analyses.

\begin{table}[]
    \caption{\textbf{Experiment result:} maximum capacity usage of sectors}
    \centering{\small
    \begin{tabular}{c c|c c c c c}
        \hline\hline
        \multicolumn{7}{c}{reallocate all file backups $100$ times} \
        \cr \cline{1-7}
        \multicolumn{2}{c}{parameter} & \multicolumn{5}{c}{maximum capacity usage} \
        \cr \cline{1-2} \cline{3-7}
        $N_{cp}$ & $N_s$ & $[1]$ & $[2]$ & $[3]$ & $[4]$ & $[5]$\\
        \hline
$10^5$ & 20 & 0.525 & 0.524 & 0.536 & 0.530 & 0.529\\\hline
$10^5$ & 100 & 0.571 & 0.566 & 0.584 & 0.572 & 0.569\\\hline
$10^6$ & 200 & 0.538 & 0.530 & 0.542 & 0.534 & 0.533\\\hline
$10^6$ & 1000 & 0.591 & 0.571 & 0.598 & 0.594 & 0.576\\\hline
$10^7$ & 2000 & 0.540 & 0.534 & 0.544 & 0.545 & 0.534\\\hline
$10^7$ & 10000 & 0.589 & 0.576 & 0.609 & 0.606 & 0.585\\\hline
$10^8$ & 20000 & 0.541 & 0.534 & 0.550 & 0.547 & 0.538\\\hline
$10^8$ & $10^5$ & 0.591 & 0.582 & 0.614 & 0.599 & 0.586\\\hline
        \hline
    \end{tabular}
    \begin{tabular}{c c|c c c c c}
        \hline\hline
        \multicolumn{7}{c}{refresh the location of a file backup $100N_{cp}$ times} \
        \cr \cline{1-7}
        \multicolumn{2}{c}{parameter} & \multicolumn{5}{c}{maximum capacity usage} \
        \cr \cline{1-2} \cline{3-7}
        $N_{cp}$ & $N_s$ & $[1]$ & $[2]$ & $[3]$ & $[4]$ & $[5]$\\
        \hline
$10^5$ & 20 & 0.532 & 0.529 & 0.538 & 0.535 & 0.531\\\hline
$10^5$ & 100 & 0.588 & 0.571 & 0.599 & 0.595 & 0.581\\\hline
$10^6$ & 200 & 0.536 & 0.535 & 0.546 & 0.542 & 0.541\\\hline
$10^6$ & 1000 & 0.592 & 0.581 & 0.610 & 0.605 & 0.589\\\hline
$10^7$ & 2000 & 0.542 & 0.535 & 0.553 & 0.549 & 0.540\\\hline
$10^7$ & 10000 & 0.610 & 0.591 & 0.626 & 0.613 & 0.599\\\hline
$10^8$ & 20000 & 0.551 & 0.547 & 0.560 & 0.558 & 0.548\\\hline
$10^8$ & $10^5$ & 0.611 & 0.604 & 0.639 & 0.628 & 0.611\\\hline
        \hline
    \end{tabular}}
    \begin{tablenotes}
      \footnotesize
      \item $[1]$: Uniform distribution in interval $[0,1]$
      \item $[2]$: Uniform distribution in interval $[1,2]$
      \item $[3]$: Exponential distribution
      \item $[4]$: Normal distribution with $\mu = \sigma^2$
      \item $[5]$: Normal distribution with $\mu = 2\sigma^2$
    \end{tablenotes}
    \label{table:experiment}
\end{table}

\subsubsection{Analysis of Robustness}
We consider the robustness of FileInsurer as the ability of resisting corruptions of sectors. The following theorem indicates that FileInsurer is quite robust. The proof is left in 
Appendix C.

\begin{restatable}{theorem}{throb}\label{th:rob}
Assume that the total size of corrupted sectors is $\lambda N_s \times minCapacity$. Denote the total value of lost files to be $V_{lost}$, and $\gamma_{lost}^v = \frac{V_{lost}}{S_v \times minValue}$ represents the ratio of the value of lost files to the total value of all files.
Then with a probability of not less than $1-c$, $\gamma_{lost}^v$ satisfies
{\small
\[
\gamma_{lost}^v \leq\max\left\{5\lambda^k,\lambda^\frac{k}{2},\frac{4\left(\frac{\log\frac{e}{2\pi}-\log c}{N_s}-\log\left(\lambda^\lambda(1-\lambda)^{1-\lambda}\right)\right)}{\gamma_v^m k \log\frac{1}{\lambda} \times capPara }\right\}.
\]
}
\end{restatable}

Let us propose a concrete example to show that the result of Theorem~\ref{th:rob} is quite strong. Set $k = 20$, $N_s = 10^6$, and $capPara = 10^3$. Let $\lambda = 0.5$, which means that half capacity of FileInsurer is broken. Then
{\small
\[
\gamma_{lost}^v\leq \max\left\{5 \times 10^{-6},0.001,
\frac{1}{\gamma^m_v} \times 5 \times 10^{-6}
\right\}.
\]
}

When $\gamma^m_v \geq 0.005$, $\gamma_{lost}^v \leq 0.001$. It means that in this case, even when half of the capacity of FileInsurer is corrupted, the value of lost files is no more than $0.1\%$ of the value of all stored files.

\subsubsection{Deposit Ratio}
The following theorem indicates that only a small deposit ratio is needed for full compensation.

\begin{restatable}{theorem}{thratio}\label{th:ratio}
Assume that the total size of corrupted sectors is no more than $\lambda N_s \times minCapacity$. If the deposit ratio satisfies
{\small
\[
\gamma_{deposit} \geq \max\left\{5 \lambda^{k-1},\lambda^{\frac{k}{2} - 1},
\frac{4}{k\times capPara }\left( \frac{\log N_s}{\log\frac{1}{\lambda}}+\frac{\log \frac{1}{c}}{\log N_s}\right)\right\},
\]
}
then full compensation can be achieved with a probability of not less than $1-c$.
\end{restatable}

The proof of \cref{th:ratio} is in 
Appendix D. 
Set $k = 20$, $N_s = 10^6$, $capPara = 10^3$ and $\lambda = 0.5$. Then $\gamma_{deposit} = 0.0046$ is enough to ensure full compensation, which is relatively small.

\subsection{Comparison with Existing Protocols}

\begin{table}
    \centering\setstretch{1.1}{\small\setlength{\tabcolsep}{1pt}
    \caption{Comparison of DSN Protocols}
    \label{table:comparison}
    \begin{tabular}{|c|c|c|c|c|c|}
        \hline
        Property & FileInsurer &Filecoin &Arweave &Storj &Sia \\
        \hline
        Capacity Scalability & Yes & Yes & Yes & Yes & Yes\\
        Preventing Sybil Attacks & Yes & Yes & Yes & Yes & No\\
        Provable Robustness & Yes & No & No & No & No\\
        Compensation for File Loss & Yes & No$^{[1]}$ & No & No & No \\
        \hline
    \end{tabular}}
    \footnotesize{$^{[1]}$ Provides only limited file loss compensation}\\
\end{table}
\Cref{table:comparison} shows the comparison between FileInsurer and existing DSN protocols including Filecoin, Arweave, Sia, and Storj. We observe that FileInsurer is the only DSN protocol that has provable robustness and gives full compensation for file loss.

\section{Discussion}\label{discussion}

In previous sections, we have proposed the general framework of FileInsurer and theoretically proved the excellent performance of FileInsurer. Besides, some practical issues exist and we explore the corresponding solutions for them under FileInsurer in this section.

\subsection{Distributions and Parameters}\label{dis:r12}
The value and size of a file follows a certain distribution in DSN. We have the following reasonable assumptions about the distribution.
\begin{itemize}
    \item The maximal value of a file is bounded by a constant. Therefore, $r_1$ (defined in ~\cref{r1}) is bounded by a constant.
    \item The average value of a unit size is a bounded constant. Then it's reasonable to assume that $\frac{\sum_f f.value}{\sum_f f.size}$ is bounded by a constant. Therefore, $r_2$ (defined in ~\cref{r2}) is bounded by a constant.
\end{itemize}

The parameters of FileInsurer should be properly set according to the distribution of files. For example, we should set parameters to make $2r_1 k$ is not far away from $r_2$ to further improve scalability bound in \cref{th:scala}. It also helps to avoid the bad situation that the total value of files is far below the maximal, but the used space has reached its limit.

\subsection{Storage Randomness When Adding or Removing Sectors}\label{maintainrandom}
In our analysis of storage randomness, we ignore the case that the network may add or remove sectors online. When a new sector $s$ is registered in the network, in order to maintain the independently and identically distributed property of the allocations, the network should traverse each allocation and swap out the allocation to that sector with the probability of $\frac{s.capacity}{N_s\times minCapacity}$. Such an operation is impossible because traversing over files is too expensive. One good approximation method is that the network first calculates how many files backups need to be swapped into the sector by sampling from a Poisson distribution, and then randomly select the file backups to swap into the sector.

If a sector is disabled, We can request it to keep storing all replicas it currently stores even if they are slowly being swapped out. As a result, it does not get easier to attack the corresponding files. When all of its files are swapped out, this sector no longer exists in the network so the storage randomness can guarantee.

\subsection{Adjusting to Extremely Large Files}
In some special cases, very few huge files, whose sizes are comparable to the capacity of sectors, need to be stored in the network. These very large files might break storage randomness because their allocations might fail to find enough space in one turn. To address this problem from the extremely large files, the network needs to specify an upper limit $sizeLimit$ on the size of a single file. For a file with a size greater than $sizeLimit$, we can convert it to a collection of segments 
by the erasure code, such that each segment's size is upper bounded by $sizeLimit$. By this operation, the file can still be recovered even if half of the segments are lost. Therefore, we can simply regard each segment as an individual file with value $\frac{2value}{k}$. In practice, we can apply the common erasure code such as Reed–Solomon code~\cite{reed1960polynomial} to archive this.

\subsection{Storing Files with Widely Varying Values}\label{sublinear}
In FileInsurer protocol, the value of each file is required to be an integer multiple of $minValue$. Thus a file with a value of $v$ can be treated as $\frac{v}{minValue}$ documents worth of $minValue$. This means that a high-value file needs to have many replicas in the system, and the number of replicas is linearly related to this file's value. A compromise solution is to pre-divide the value levels of files and to establish a storage subnetwork corresponding to each level. Then the clients can choose which subnetwork to store files based on the value level of their files.

\subsection{Avoiding Selfish Storage Providers}
Selfish storage providers refer to these providers who store files but do not normally provide retrieval services. Assume the ratio of the number of selfish storage providers to the number of all providers is $\alpha$ in the network. Then it is expected that $\alpha^k$ proportion of files suffer from the threat of the selfish providers' collusion. Here $k$ is just the number of copies of a stored file. 
As a result, any protocol that fixes file storage locations cannot fundamentally solve the problem of selfish storage providers. However, a natural advantage of FileInsurer is that its file refresh mechanism can fundamentally eliminate the threat from selfish storage providers. Because of the existence of refreshing file storage location, no single file will be completely controlled by the selfish storage provider for a long time.

\subsection{Supports for IPFS}
Filecoin has shown how to support IPFS in a blockchain-based DSN, and FileInsurer has a similar approach. In FileInsurer, the hashes and locations of files are all stored in blockchain. Therefore, it's easy to build and update DHTs and Merkle DAGs on FileInsurer so that anyone can address files stored in FileInsurer through IPFS paths. The retrieval of files can be also realized through BitSwap protocol.

\section{Conclusion}\label{conclusion}

In this paper, we propose FileInsurer, a novel design for blockchain-based \emph{Decentralized Storage Network}, which achieves both scalability and reliability. FileInsurer is the first DSN protocol that gives full compensation to file loss and has provable robustness. Our work also raises many open problems. First, are there other approaches to enhance the reliability of \emph{Decentralized Storage Networks}? For example, a reputation mechanism~\cite{chen2021provable} on storage providers may be also helpful to reduce the loss of files. Second, are there other ways to support dynamic content in sectors other than DRep? Furthermore, can the idea of FileInsurer be extended to decentralized insurance in other scenarios?

\bibliographystyle{abbrv}
\bibliography{references}

\newpage
\begin{appendices}
\section{Proof of Theorem 1}
\thscala*
\begin{proof}
There are two restrictions on the total size of raw files. One is the restriction of total capacity. The other is the restriction of the maximal total value of stored files.

Under the former restriction, each file $f$ is stored as $f.cp$ replicas. Due to the assumption of redundant capacity, The total size of all replicas can not exceed $\frac{1}{2}$ of total capacity. That is,
\[
\sum_f (f.size \times f.cp) \leq \frac{1}{2} N_s \times minCapacity.
\]
Because $f.cp = k\times \frac{f.value}{minValue}$, we have
\[
\sum_f (f.size \times k\times \frac{f.value}{minValue}) \leq \frac{1}{2} N_s \times minCapacity.
\]
Then we have
\[
k \times\frac{\sum_f (f.size \times f.value)}{minValue \times \sum_f f.size}  \leq \frac{1}{2} \times \frac{N_s \times minCapacity}{\sum_f f.size}.
\]
Let $r_1 = \frac{\sum_f f.size \times f.value}{minValue \times \sum_f f.size}$  we have
\[
\frac{k}{r_1}  \leq \frac{1}{2} \times \frac{N_s \times minCapacity}{\sum_f f.size}.
\]
Then
\[
\sum_f f.size \leq \frac{ N_s \times minCapacity}{2r_1 k }.
\]

Under the latter restriction, the total value of files can't exceed $S^m_v \times minValue$. That is,
\[
\sum_f f.value \leq n \times minValue.
\]
Because $n = Cap\_Para\times m$,
\[
\frac{\sum_f f.value}{\sum_f f.size} \leq \frac{Cap\_Para\times N_s  \times minValue}{\sum_f f.size}.
\]
Therefore,
\[
\sum_f f.size \leq \frac{N_s \times minCapacity }{r_2}, 
\]
where
\[
r_2 = \frac{minCapcity \times \sum_f f.value}{minValue \times \sum_f f.size \times Cap\_Para}.
\]
\end{proof}

\section{Proof of Theorem 2}
\thadd*
\begin{proof}
In the special case, all files have the same size $f.size$. For a sector $s$ with capacity $s.capacity$, it can store $\frac{s.capacity}{f.size}$ backups. We define $N_{cp}=kN_v$ as the number of file backups in total because $N_{cp} = \sum_f f.cp = \sum_f \frac{f.value}{minValue} \times k = kN_v$. Additionally, let $X_i$ be the event that the backup $i$ is stored in this sector and $S=\sum_{i=1}^{N_{cp}}X_i$. Because the assumption of redundant capacity, we have $\mathrm{E}[S]\leq \frac{s.capacity}{2 f.size}$. By multiplicative Chernoff bound, we have
\begin{align*}\small
& \Pr\left[s.freeCap\leq \frac{1}{8} s.capacity\right]\\
= & \Pr\left[\sum_{i=1}^{N_{cp}} X_i\geq \frac{7}{8}\frac{s.capacity}{f.size}\right]\\
\leq &  \Pr\left[S\geq \frac{7}{4}\mathrm{E}\left[S\right]\right]\\
\leq & \exp\left\{\left(\log\frac{e}{4}\right)\frac{3}{4}\mathrm{E}\left[S\right]\right\}\\
\leq & \exp\left\{\left(\log\frac{e}{4}\right)\frac{3 s.capacity}{8 f.size}\right\}\\
\leq & \exp\left\{-0.144\frac{s.capacity}{f.size}\right\}
\end{align*}

By applying union bound, we obtain
\small{\[
\Pr\left[ \exists s, s.freeCap\leq \frac{1}{8}s.capacity\right] \leq N_s\exp\left\{-0.144\frac{s.capacity}{f.size}\right\}.
\]}
\end{proof}

\section{Proof of Theorem 3}
We define the state of a FileInsurer network as $(F,S,A,C)$ consisting of files $F$ ,sectors $S$, all allocations $A$, and corrupted bits $C$ in the network. Also, we define $V_{lost}^{(F,S,A,C)}$ as the sum of values of the lost files and $V_{confiscated}^{(F,S,A,C)}$ as the confiscated deposits of the corrupted sectors.

\begin{lemma}
For a specific state $(F,S,A,C)$, keeping the content and availability of each physical disk in the network unchanged, it can be viewed as another state $(F',S,A',C)$ where the value of each file is $minValue$. State $(F',S,A',C)$ satisfies $V_{lost}^{(F,S,A,C)}\leq V_{lost}^{(F',S,A',C)}$.
\label{lemma:filesize}
\end{lemma}
\begin{proof}
We divide each file descriptor $f$ to $\frac{f.value}{minValue}$ different file descriptors. These file descriptors all have the value $minValue$ and the same Merkle root as that of $f$. Divide the $f.cp$ allocations of $f$ equally among these new file descriptors, so each file descriptor have exactly $k$ allocations. By defining $F'$ as all these new file descriptors, $A'$ as these new file allocations, we construct a state $(F',S,A',C)$ such that the value of each file is $minValue$.

The content and availability of each physical disk in the network are same in state $(F,S,A,C)$ and state $(F',S,A',C)$. Since we do not change the state of sectors, we simply obtain $V_{confiscated}^{(F,S,A,C)}= V_{confiscated}^{(F',S,A',C)}$. For each file $f$ lost in state $(F,S,A,C)$, since all its backups are lost, every new file descriptor generated by $f$ in state $(F',S,A',C)$ also lost. The value of the file $f$ is equal to the sum of the values of the file descriptors it generates, so $V_{lost}^{(F,S,A,C)}\leq V_{lost}^{(F',S,A',C)}$.
\end{proof}

\begin{lemma}
$\forall 0<p\leq \frac{1}{5}$ and $5p\leq x\leq 1$, $D_{KL}(x||p)\geq \frac{1}{2} x\log\frac{x}{p}$.
\label{lemma:1}
\end{lemma}
\begin{proof}
In the proof below, we will use $x\geq p$ unspecified. Let 
\[
\begin{cases}
f(x)=x^2\log\frac{x}{p}-(1-x)^2\log\frac{1-x}{1-p}\\
g(x)=\log\frac{x-1}{p-1} \left(-\log\frac{x}{p}+x-1\right)-x \log\frac{x}{p}\\
h(x)=\frac{x\log\frac{x}{p}}{-(1-x)\log\frac{1-x}{1-p}}
\end{cases}
\]whose domain is $x\in[p,1]$. First, we have $f(x)\geq 0$ because $\frac{\mathrm{d}f}{\mathrm{d}x}=1+2D_{KL}(x||p)\geq 0$ and $f(p)=0$. Then, we have $g(x)\geq 0$ because $\frac{\mathrm{d}g}{\mathrm{d}x}=\frac{f(x)}{x(1-x)}\geq 0$ and $g(p)=0$. Finally, we obtain $h(x)$ is monotonically increasing because $\frac{\mathrm{d}h}{\mathrm{d}x}=\frac{g(x)}{(x-1)^2 \log ^2\frac{1-x}{1-p}}\geq 0$. 

Because $h(x)$ is monotonically increasing, $\forall x\geq 5p$, $ h(x)\geq h(5p)=\frac{5 p \log 5}{(1-5p) \log\frac{1-p}{1-5p}}$. We can find that
\begin{align*}
\frac{\mathrm{d}}{\mathrm{d}p} \frac{5 p \log 5}{(1-5p) \log\frac{1-p}{1-5p}}
= & \frac{5 \log 5 \left((1-p) \log\frac{1-p}{1-5 p}-4 p\right)}{(1-5 p)^2 (1-p) \log ^2\left(\frac{1-p}{1-5 p}\right)}\\
\geq & \frac{5 \log 5 \left((1-p) (1-\frac{1-5p}{1-p})-4 p\right)}{(1-5 p)^2 (1-p) \log ^2\left(\frac{1-p}{1-5 p}\right)}\\
= & 0.
\end{align*}
When $p\rightarrow 0$, $\frac{5 p \log 5}{(1-5p) \log\frac{1-p}{1-5p}}>2$, so $\forall 0< p\leq \frac{1}{5},\frac{5 p \log 5}{(1-5p) \log\frac{1-p}{1-5p}}>2$, that is $\forall x\geq 5p,h(x)>2$. At last,
\begin{align*}
D_{KL}(x||p)&=  x\log\frac{x}{p} + (1-x)\log\frac{1-x}{1-p}\\
&=  x\log\frac{x}{p} \left (1+\frac{(1-x)\log\frac{1-x}{1-p}}{x\log\frac{x}{p}}\right)\\
&=  x\log\frac{x}{p} \left (1-\frac{1}{h(x)}\right)\\
&\geq  \frac{1}{2} x\log\frac{x}{p}.
\end{align*}
\end{proof}

\begin{lemma}
Assume that the total size of corrupted sectors is $\lambda N_s \times minCapacity$. Denote the total value of lost files to be $V_{lost}$. Then, with a probability of not less than $1-c$, $V_{lost}$ satisfies
{\small
\[
V_{lost} \leq minValue \times \max\left\{5N_v\lambda^k,N_v\lambda^\frac{k}{2},4\frac{\log\binom{N_s}{\lambda N_s}-\log c}{k\log\frac{1}{\lambda}}\right\}.
\]
}
\label{lemma:2}
\end{lemma}

\begin{proof}
Because of \cref{lemma:filesize}, we can make the relaxation of that each file has value $minValue$ in subsequent analysis. Under the relaxations, the setting of the problem can be simplified as follow: there are $N_v$ files, each file has the same value $minValue$ and needs to be stored in $k$ sectors. Each storage location of each file is generated independent and identically distributed. 

For any certain scheme that the adversary corrupts $\lambda$ ratio of capacity, which means that the total size of corrupted sectors is $\lambda N_s \times minCapacity$, define random variable $X_i$ as the indicator variable of that the file $f_i$ is lost. Recall \emph{storage randomness} indicates that all replicas are evenly and randomly distributed. Then all $X_i$ are independent events and $\Pr[X_i]=\lambda^{k}$ when any sectors with total space of $\lambda N_s \times minCapacity$ are corrupted.

Denote $\gamma = \frac{V_{lost}^v}{minValue}$.
By Chernoff bound and \cref{lemma:1}, when $\gamma \geq 5N_v\lambda^k $, we obtain 
\begin{align*}
& \Pr\left[\sum_iX_i \geq \gamma \right]\\
\leq & \exp\left\{-N_v\left(\frac{\gamma}{N_v}\log\frac{\gamma}{N_v\lambda^k}+\left(1-\frac{\gamma}{N_v}\right)\log\frac{N_v-\gamma}{N_v-N_v\lambda^k}\right)\right\}\\
\leq & \exp\left\{-\frac{\gamma}{2}\log\frac{\gamma}{N_v\lambda^k}\right\}.
\end{align*}

Consider the number of scenarios in which an adversary can corrupt sectors with capacity $\lambda N_s \times minCapacity$. Here we do a simple scaling that treat a sector with $s.capacity$ capacity as $\frac{s.capacity}{minCapacity}$ sectors with capacity $minCapacity$. After this scaling, the adversary has $\binom{N_s}{\lambda N_s}$ options to corrupt sectors. Therefore, for the original situation, the number of scenarios in which an adversary can corrupt sectors with capacity $\lambda N_s \times minCapacity$ does not exceed $\binom{N_s}{\lambda N_s}$.

By union bound, when $\gamma\geq 5N_v\lambda^k$, the probability it cannot manufacture $\gamma$ lost files is at least
\[
1-\binom{N_s}{\lambda N_s}\exp\left\{-\frac{\gamma}{2}\log\frac{\gamma}{N_v\lambda^k}\right\}.
\]

When $\gamma\geq N_v\lambda^\frac{k}{2}$, we have $\log\frac{\gamma}{N_v\lambda^k}\geq \log\left(\lambda^{\frac{-k}{2}}\right)=-\frac{k}{2}\log\lambda$. Then we find that
\begin{align*}
& \gamma  \geq  4\frac{\log\binom{N_s}{\lambda N_s}-\log c}{k\log\frac{1}{\lambda}}\\
\Leftrightarrow & \gamma\frac{k}{4}\log\frac{1}{\lambda}  \geq  -\log\frac{c}{\binom{N_s}{\lambda N_s}}\\
\Rightarrow & \frac{\gamma}{2}\log\frac{\gamma}{N_v\lambda^k}  \geq  -\log\frac{c}{\binom{N_s}{\lambda N_s}}\\
\Leftrightarrow & -\frac{\gamma}{2}\log\frac{\gamma}{N_v\lambda^k}  \leq  \log\frac{c}{\binom{N_s}{\lambda N_s}}\\
\Leftrightarrow & \exp\left\{-\frac{\gamma}{2}\log\frac{\gamma}{N_v\lambda^k}\right\}  \leq  \frac{c}{\binom{N_s}{\lambda N_s}}\\
\Leftrightarrow & \binom{N_s}{\lambda N_s}\exp\left\{-\frac{\gamma}{2}\log\frac{\gamma}{N_v\lambda^k}\right\}  \leq  c.
\end{align*}

This shows that when $\gamma$ meets the above three conditions, the probability that an adversary can make $\gamma$ lost files does not exceed $1-c$, that is, $\gamma$ needs to satisfy
\[
\gamma \leq\max\left\{5\lambda^kN_v,\lambda^\frac{k}{2}N_v,4\frac{\log\binom{N_s}{\lambda N_s}-\log c}{k\log\frac{1}{\lambda}}\right\}.
\]

Therefore, with a probability of no less than $1-c$, $V_{lost}^v$ satisfies
{\small
\[
V_{lost}^v \leq minValue \times \max\left\{5N_v\lambda^k,N_v\lambda^\frac{k}{2},4\frac{\log\binom{N_s}{\lambda N_s}-\log c}{k\log\frac{1}{\lambda}}\right\}.
\]
}

\end{proof}

\throb*
\begin{proof}
Now we use an upper bound of the binomial number via Stirling's formula,
\begin{align*}
\binom{N_s}{\lambda N_s} & = \frac{N_s!}{(\lambda N_s)!(N_s-\lambda N_s)!}\\
& \leq \frac{e}{2\pi}\frac{N_s^{N_s+\frac12}}{(\lambda N_s)^{\lambda N_s+\frac12}(N_s-\lambda N_s)^{N_s-\lambda N_s+\frac12}}\\
& = \frac{e}{2\pi}\sqrt{\frac{1}{N_s\lambda(1-\lambda)}}\left(\frac{1}{\lambda^\lambda(1-\lambda)^{1-\lambda}}\right)^{N_s}\\
& \leq \frac{e}{2\pi}\left(\frac{1}{\lambda^\lambda(1-\lambda)^{1-\lambda}}\right)^{N_s}
\end{align*}

Using this upper bound, we can have a simpler version of $\gamma_{lost}^v$:
\begin{align*}
\gamma_{lost}^v & \leq\max\left\{5\lambda^k,\lambda^\frac{k}{2},4\frac{\log\frac{e}{2\pi}-N_s\log\left(\lambda^\lambda(1-\lambda)^{1-\lambda}\right)-\log c}{N_vk\log\frac{1}{\lambda}}\right\}\\
& = \max\left\{5\lambda^k,\lambda^\frac{k}{2},4\frac{\frac{\log\frac{e}{2\pi}-\log c}{N_s}-\log\left(\lambda^\lambda(1-\lambda)^{1-\lambda}\right)}{capPara\cdot \gamma_v^m k\log\frac{1}{\lambda}}\right\}\\
\end{align*}
\end{proof}

\section{Proof of Theorem 4}
\thratio*
\begin{proof}
By assumption, the total size of corrupted sectors is no more than $\lambda N_s \times minCapacity$.
Because the deposit of corrupted sectors should always cover the file loss, for all $\frac{1}{N_s}\leq \lambda^{\prime} \leq \lambda$ we shall have $\lambda^{\prime} \gamma_{deposit}  N_v^m \geq \gamma$, which is equivalent to
\[
\gamma_{deposit} \geq  \max_{\frac{1}{N_s}\leq\lambda^{\prime}\leq \lambda} \left\{ \frac{\gamma}{\lambda^{\prime} N_v^m}\right\}.
\]
Then with probability no less than $1-c$, the following $\gamma_{deposit}$ is enough for full compensation
{\small
\[
\gamma_{deposit} \geq \max_{\frac{1}{N_s}\leq\lambda^{\prime}\leq \lambda} \max \left\{5(\lambda^{\prime})^{k-1},(\lambda^{\prime})^{\frac{k}{2}-1},4\frac{\log\binom{N_s}{\lambda^{\prime} N_s}-\log c}{\lambda^{\prime}N_v^m k\log\frac{1}{\lambda^{\prime}}}\right\}.
\]
}

Considering the third part, we have
\begin{align*}
& 4\frac{\log\binom{N_s}{\lambda^{\prime} N_s}-\log c}{\lambda^{\prime}N_v^m k\log\frac{1}{\lambda^{\prime}}}\\
\leq & \max_{\frac{1}{N_s}\leq\lambda^{\prime}\leq \lambda}
\frac{4\log\left(N_s^{\lambda^{\prime} N_s}\right)-4\log c}{\lambda^{\prime} N_v^mk\log\frac{1}{\lambda^{\prime}}}\\
= & \max_{\frac{1}{N_s}\leq\lambda^{\prime}\leq \lambda}
\frac{4\lambda^{\prime} N_s\log N_s-4\log c}{\lambda^{\prime} N_v^mk\log\frac{1}{\lambda^{\prime}}}\\
\leq & \left(\max_{\frac{1}{N_s}\leq\lambda^{\prime}\leq \lambda}\frac{4N_s\log N_s}{N_v^mk\log\frac{1}{\lambda^{\prime}}}\right)+\left(\max_{\frac{1}{N_s}\leq\lambda^{\prime}\leq \lambda}\frac{-4\log c}{\lambda^{\prime} N_v^m k\log\frac{1}{\lambda^{\prime}}}\right)\\
\leq & \frac{4N_s\log N_s}{N_v^mk\log\frac{1}{\lambda}}+\frac{-4N_s\log c}{N_v^mk\log N_s}.
\end{align*}
Then
{\small
\[
\gamma_{deposit} \geq \max\left\{5 \lambda^{k-1},\lambda^{\frac{k}{2} - 1},
\frac{4N_s\log N_s}{N_v^mk\log\frac{1}{\lambda}}+\frac{-4N_s\log c}{N_v^mk\log N_s}\right\}.
\]
}
As $capPara = \frac{N_v^m}{N_s} $,
{\small
\[
\gamma_{deposit} \geq \max\left\{5 \lambda^{k-1},\lambda^{\frac{k}{2} - 1},
\frac{4}{k\times capPara }\left( \frac{\log N_s}{\log\frac{1}{\lambda}}+\frac{\log \frac{1}{c}}{\log N_s}\right)\right\}.
\]
}

\end{proof}

\end{appendices}

\end{document}